\documentclass[11pt]{article}

\usepackage{dima}

\newcommand{\MC}{\textsc{Max-Cut}\xspace}

\newif\iffigfromfile
\figfromfilefalse
\figfromfiletrue %

\begin{document}

\title{Sketching Cuts in Graphs and Hypergraphs}

\author{Dmitry Kogan%
\thanks{Work supported in part by a US-Israel BSF grant \#2010418,
Israel Science Foundation grant \#897/13,
and by the Citi Foundation.
Email: \texttt{\{dmitry.kogan,robert.krauthgamer\}@weizmann.ac.il}
}
\qquad\qquad Robert Krauthgamer%
\footnotemark[1]
\\
Weizmann Institute of Science
}

\maketitle

\begin{abstract}
  Sketching and streaming algorithms are in the forefront of current
  research directions for cut problems in graphs.
  In the streaming model, we show that $(1-\eps)$-approximation
  for \MC must use $n^{1-O(\eps)}$ space; moreover,
  beating $4/5$-approximation requires polynomial space.
  For the sketching model, we show that $r$-uniform hypergraphs
  admit a $(1+\eps)$-cut-sparsifier (i.e., a weighted subhypergraph
  that approximately preserves \emph{all} the cuts)
  with $O(\eps^{-2}n(r+\log n))$ edges.
  We also make first steps towards sketching general CSPs 
  (Constraint Satisfaction Problems).
\end{abstract}

\section{Introduction}

The emergence of massive datasets has turned many algorithms
impractical, because the standard assumption of having (fast) random
access to the input is no longer valid.  One example is when data is
too large to fit in the main memory (or even on disk) of one machine;
another is when the input can be accessed only as a stream, e.g.,
because its creation rate is so high, that it cannot even be stored in
full for further processing.
Luckily, the nature of the problems has evolved too, and we may often
settle on approximate, rather than exact, solutions.

These situations have led to the rise of new computational paradigms.
In the \emph{streaming model} (aka \emph{data-stream}), the input can
be accessed only as a stream (i.e., a single pass of sequential
access), and the algorithm's space complexity (storage requirement)
must be small relative to the stream size.  In the \emph{sketching
  model}, the input is summarized (compressed) into a so-called
sketch, which is short and suffices for further processing without
access to the original input.
The two models are related -- sketches are often useful in the design
of streaming algorithms, and vice versa.  In particular, lower bounds
for sketch-size often imply lower bounds on the space complexity of
streaming algorithms.

\paragraph{Graph problems.}
Recently, the streaming model has seen many exciting developments on
\emph{graph problems}, where an input graph $G=(V,E)$ is represented
by a stream of edges.  The algorithm reads the stream and should then
report a solution to a predetermined problem on $G$, such as graph
connectivity or maximum matching; see e.g.\ the
surveys~\cite{Zhang10,McGregor14}.
Throughout it will be convenient to denote $n=\card{V}$, and to assume
edges have weights, given by $w:E\to\R_+$.
While initial efforts focused on polylogarithmic-space algorithms,
various intractability results have shifted the attention to what is
called the \emph{semi-streaming} model, where the algorithm's space
complexity is $\tilde O(n)$.%
\footnote{We use $\tilde O(f)$ to denote $O(f\polylog f)$,
which suppresses logarithmic terms.
}
In general, this storage is not sufficient to record the entire edge-set.

Cuts in graphs is a classical topic that has been studied extensively
for more than half a century,
and the last two decades have seen a surge of attention turning
to the question of their succinct representation.
The pioneering work of Bencz\'ur and Karger~\cite{benczurkarger1996}
introduced the notion of \emph{cut sparsifiers}:
given an undirected graph $G=(V,E,w)$,
a $(1+\eps)$-sparsifier is a (sparse) weighted subgraph $G'=(V,E',w')$
that preserves the value of every cut up to a multiplicative factor $1+\eps$.
Formally, this is written as
$$
  \forall S\subset V,
  \qquad
  1 \leq \frac{w'(S,\bar S)}{w(S,\bar S)} \leq 1+\eps;
$$
but it is sometimes convenient to replace the lefthand-side with $1-\eps$
or $\frac{1}{1+\eps}$, which affects $\eps\leq \tfrac12$ by only a
constant factor.  In addition to their role in saving storage,
sparsifiers are important because they can speed-up graph algorithms
whose running time depends on the number of edges.  Observe that
sparsifiers are a particularly strong form of graph-sketches since on
top of retaining the value of all cuts, they hold the additional
property of being subgraphs, rather than arbitrary data structures.

Ahn and Guha~\cite{ag09} built upon the machinery of cut sparsifiers
to present an $\tilde O(n/\eps^2)$-space streaming algorithm that
can produce a $(1+\eps)$-approximation to all cuts in a graph.
Further improvements handle also edge deletions \cite{agm12a, agm12, gkp12},
or the stronger notion of spectral sparsification (see~\cite{KLMMMS14}
and references therein).
These results are nearly optimal, due to a space lower bound
of $\Omega(n/\eps^2)$ for sketching all cuts in a graph \cite{AKW14}
(which improves an earlier bound of \cite{ag09}).

\paragraph{Recent Directions.}
These advances on sketching and streaming of graph cuts inspired new
questions.  One direction is to seek space-efficient streaming
algorithms for \emph{specific cut problems}, such as approximating
\MC, rather than \emph{all} cuts.  A second direction concerns
\emph{hypergraphs}, asking whether cut sparsification, sketching and
streaming can be generalized to hypergraphs.  Finally, viewing cuts in
graphs and hypergraphs as special cases of \emph{constraint
  satisfaction problems} (CSPs), we ask whether other CSPs also admit
sketches.
Currently, there is a growing interest in generalizing graph cut
problems to broader settings, such as sparsifying general \emph{set
  systems} using small weighted samples \cite{NR13},
\emph{high-dimensional expander theory} \cite{KKL14}, sparsest-cuts in
hypergraphs \cite{LM14,Louis14}, and applications of hypergraph cuts in
\emph{networking} \cite{yamaguchicyber}.

\subsection{Our Results}

We first address a natural question raised in \cite[Question
10]{OpenProblemsStreaming11}, whether the well-known \MC problem
admits approximation strictly better than factor $1/2$
by streaming algorithms that use space sublinear in $n$.
Here, \MC denotes the problem of computing the \emph{value} of a
maximum cut in the input graph $G$ (and not the cut itself),
since \emph{reporting} a cut requires space $\Omega(n)$
(see Subsection~\ref{lb_actual_cut} for a short proof).
We prove that for every fixed $\eps\in(0,\tfrac15)$, streaming
algorithms achieving $(1-\eps)$-approximation for \MC must use
$n^{1-O(\eps)}$ space.
In fact, even beating $4/5$-approximation requires polynomial space.
Our result is actually stronger and holds also
in a certain sketching model.
Previously, it was known that streaming computation of $\MC$ \emph{exactly}
requires $\Omega(n^2)$ bits \cite{Zelke11}.
Our proof is by reduction from the \textsc{Boolean Hidden Hypermatching} problem,
and captures the difficulty of distinguishing, under limited communication,
whether the graph is avertex-disjoint union of even-length cycles
(in which case the graph is bipartite) or of odd-length cycles
(in which case we can bound the maximum cut value).
See Section \ref{sec:MC} for details.

\medskip
Second, we study sparsification of cuts in \emph{hypergraphs},
and prove that every $r$-uniform hypergraph admits a sparsifier
(weighted subhypergraph) of size $\tilde O(rn/\eps^2)$
that approximates all cuts within factor $1\pm\eps$.
This result immediately implies sketching and streaming algorithms
(following \cite{ag09}).
Here, the weight of cut $(S,\bar S)$ in a hypergraph $H=(V,E,w)$ is
the total weight of all hyperedges $e\in E$ that intersect both $S$
and $\bar S$.%
\footnote{Another possible definition, see \cite[Corollary 7]{dCHS11},
  is $\sum_{e\in E} w_e\cdot \card{e\cap S}\cdot\card{e\cap \bar S}$.
  The latter definition seems technically easier for sparsification,
  although both generalize the case of ordinary graphs ($r=2$).
}
This question was raised by
de Carli Silva et al.~\cite[Corollary 8]{dCHS11}, who show that every
$r$-uniform hypergraph has a sparsifier of size $O(n)$ that
approximates all cuts within factor $\Theta(r^2)$. As hypergraph
cuts can be viewed as set-systems, sparsifiers of size $O(n^2)$,
regardless of $r$, follow implicitly from \cite{NR13}.
Along the way, we establish interesting, if not surprising, bounds
on the number of approximately minimum cuts in hypergraphs.
Technically, this is our most substantial contribution,
see Section~\ref{sec:hypergraphs} for details.

\medskip
Finally, as a step towards understanding a wider range of
CSPs, we show that $k$-SAT instances on $n$ variables admit a sketch
of size $\tilde O(kn/\eps^2)$ that can be used to
$(1+\eps)$-approximate the value of all truth assignments.
We prove this result in Section~\ref{sec:SAT} by reducing it
to hypergraph sparsification.
We remark that sketching of SAT formulae was studied in a different setting,
where some computational-complexity assumptions were used
in~\cite{vanmelkebeek} to preclude a significant size-reduction
that \emph{preserves the satisfiability} of the formula.
Our sparsification result
differs in that it \emph{approximately preserves the value} of all
assignments.

\paragraph{Related Work.}
Independently of our work, Kapralov, Khanna and Sudan~\cite{KKS14}
studied the same problem of approximating \MC in the streaming
model. They first prove that for every fixed $\eps>0$,
streaming algorithms achieving $(1-\eps)$-approximation for \MC must use
$n^{1-O(\eps)}$ space.
(This is similar to our Theorem~\ref{thm:streaming_mc}.)
They then make significant further progress,
and show that achieving an approximation ratio strictly better than
(the trivial) $1/2$ require $\tilde\Omega(\sqrt n)$ space.
In fact, this result holds even if the edges of the graph
are presented in a random (rather than adversarial) order.

\section{Sketching \MC}
\label{sec:MC}

The classical \MC problem is perhaps the simplest \textsc{Max-CSP} problem.
Therefore, it has been studied extensively,
leading to fundamental results both in approximation algorithms~\cite{GW95}
and in hardness of approximation~\cite{KKMO04}.
It is thus natural to study \MC also in the streaming model.
As mentioned above, preserving the values of \emph{all} cuts in a graph requires linear space even if only approximate values are required
\cite{ag09,AKW14},
which raises the question whether smaller space suffices to approximate
only the \MC value (as mentioned above,
it is natural to require the algorithm to report
only the value of the cut as opposed to the cut itself,
see Section~\ref{lb_actual_cut}).

Sketching all cuts in a graph clearly preserves also the maximum-cut
value, and thus an $\tilde O\left(\frac n{\eps^2}\right)$ space
streaming algorithm for a $(1-\eps)$-approximation of \MC follows
immediately from \cite{ag09}. Yet since the maximum cut value is
always $\Omega(m)$, where $m$ is the total number (or weight) of all
edges, a similar result can be obtained more easily by uniform
sampling (achieving $\eps m$ additive approximation for all cuts)
\cite[Theorem 21]{Zelke09}. The latter approach has the additional
advantage that it immediately extends to hypergraphs.

It turns out that this relatively straightforward
approach is not far from optimal,
as we prove that streaming algorithms that give a
$(1-\eps)$-approximation for \MC require $n^{1-O(\eps)}$ space.
\begin{theorem}\label{thm:streaming_mc}
  Fix a constant $\eps\in(0,\tfrac 15)$.
  Every (randomized) streaming
  algorithm that gives a $(1-\eps)$-approximation of the \MC value
  in $n$-vertex graphs requires
  space $\Omega(n^{1-1/t})$ for $t= \lfloor \frac 1{2\eps} - \frac 12 \rfloor$,
  which in particular means space $n^{1-O(\eps)}$.
\end{theorem}
To prove this result, we consider the somewhat stronger \emph{one-way
  two-party communication model}, where instead of arriving as a
stream, the set of edges of a graph is split between two parties, who
engage in a communication protocol to compute (approximately)
the graph's maximum-cut value.
Since a lower bound in this model
immediately translates to the original streaming model,
the theorem above follows immediately from Theorem~\ref{thm:oneway_lb} below.

\subsection{Proof of Theorem \ref{thm:streaming_mc}}

\begin{definition}[$\MC^\eps$] Let $G=(V,E_A\cup E_B)$ be an input
  graph on $|V|=n$ vertices with maximum cut value\footnote{For the
    proof of the lower bound it suffices to restrict our attention to
    unweighted graphs, with all edges having unit weight.} $c^*$, and
  $\eps>0$ some small constant. $\MC^\eps$ is a two player
  communication game where Alice and Bob receive the edges $E_A$ and
  $E_B$ respectively and need to output a value $c'$ such that with
  high probability $(1-\eps)c^* \le c' \le c^*$.
\end{definition}

\begin{theorem} \label{thm:oneway_lb}
  Fix a constant $\eps \in (0,\tfrac 15)$.
  Then the randomized one-way communication complexity of $\MC^\eps$ is
  $\Omega(n^{1-1/t})$ for $t=\lfloor \frac 1{2\eps}-\frac 12\rfloor$.
\end{theorem}

The proof is by a reduction from the following communication problem
studied in \cite{Verbin2011}.
\newcommand{\BHH}{\textsc{BHH}}
\begin{definition}[$\BHH_n^t$]
  The \textsc{Boolean Hidden Hypermatching} problem is a communication
  complexity problem where
  \begin{itemize}\compactify
  \item Alice gets a boolean vector $x\in\bits^n$
    where $n = 2kt$ for some integer $k$,
  \item Bob gets a perfect hypermatching $M$ on $n$ vertices where each
  edge has $t$ vertices and a boolean vector $w$ of length $n/t$.
  \end{itemize}
  Let $Mx$ denote the length-$n/t$ boolean vector $(\bigoplus_{1\le
    i\le t} x_{M_{1,i}},\dots,\bigoplus_{1\le i\le t}x_{M_{n/t,i}})$
  where $(M_{1,1},\dots,M_{1,t})$ $,\dots,$$(M_{n/t,1},\dots,M_{n/t,t})$
  are the edges of M. It is promised that either $Mx\oplus w =
  1^{n/t}$ or $Mx\oplus w = 0^{n/t}$. The problem is to return $1$ in
  the former case, and to return $0$ in the latter.
\end{definition}

\begin{lemma}[{\cite[Theorem 2.1]{Verbin2011}}] The randomized one-way
  communication complexity of $\BHH_n^t$ where $n=2kt$ for some integer
  $k\ge1$ is $\Omega(n^{1-1/t})$.  \label{lem:bhh_lb}
\end{lemma}

\begin{proof}[Proof of Theorem \ref{thm:oneway_lb}]
  We show a reduction from $\BHH_n^t$ to $\MC^\eps$.
  Consider an instance $(x,M,w)$ of the $\BHH_n^t$ problem: Alice gets
  $x\in\bits^n$, and Bob gets a perfect hypermatching $M$ and a
  vector $w\in\bits^{n/t}$.

  We construct a graph $G$ for the $\MC^\eps$ problem as follows (see
  Figure \ref{fig:bhh_mc_red} for an example):
  \begin{itemize}
  \item The vertices of $G$ are
    $V=\{v_i\}_{i=1}^{2n}\cup\{u_i\}_{i=1}^{2n}\cup\{w_i\}_{i=1}^{2n/t}$.
  \item The edges $E_A$ given to Alice are: for every $i\in[n]$, if
    $x_i=0$, Alice is given two ``parallel'' edges
    $(u_{2i-1},v_{2i-1}),(u_{2i},v_{2i})$; if $x_i=1$, Alice is
    given two ``cross'' edges $(u_{2i-1},v_{2i}),(u_{2i},v_{2i-1})$.
  \item The edges $E_B$ given to Bob are: for each hyperedge
    $M_j=(i_1,i_2,\dots,i_t)\in M$ (where the order is fixed
    arbitrarily):
    \begin{itemize}
    \item For $k=1,2,\dots,t-1$, Bob is given
      $(u_{2i_k-1},v_{2i_{k+1}-1})$ and $(u_{2i_k},v_{2i_{k+1}})$
    \item For $k=t$, Bob is given
      $(u_{2i_t},w_{2j})$ and $(v_{2i_t-1},w_{2j-1})$;
    \item If $w_j=0$ Bob is given two ``parallel'' edges
      $(w_{2j},v_{2i_1})$ and $(w_{2j-1},v_{2i_1-1})$; if $w_j=1$, Bob
      is given two ``cross'' edges $(w_{2j},v_{2i_1-1})$ and
      $(w_{2j-1},v_{2i_1})$
    \end{itemize}
  \end{itemize}

  By definition, for each $j\in[n/t]$, if $M_j=(i_1,i_2,\dots,i_t)\in
  M$ and $(Mx)_j\oplus w_j=0$ we have $\sum_{k=1}^tx_{i_k}\oplus
  w_j=0$. Since the number of $1$ bits in the latter sum is even, when
  we start traversing from $u_{2i_1}$ we go through an even number of
  ``cross'' edges and complete a cycle of length $2t+1$. Similarly
  when starting our traversal at $u_{2i_1-1}$ we complete a different
  cycle of the same length. Therefore if $(x,M,w)$ is a $0$-instance
  the graph consists of $\frac{2n}t$ paths of (odd) length $2t+1$
  each. Therefore the maximum cut value is $c^*_0=2t\cdot\frac {2n}t =
  4n$.

  On the other hand if $(Mx)_j\oplus w_j=1$, starting our traversal at
  $u_{2i_1-1}$, we pass an odd number of cross edges and end up at
  $u_{2i_1}$, from where we once again pass an odd number of cross
  edges, to complete a cycle of total length $2\cdot(2t+1)=4t+2$ that
  ends back in $u_{2i_1-1}$. Therefore, if $(x,M,w)$ is a $1$-instance
  the graph consists of $n/t$ paths of (even) length $4t+2$ each. The
  maximum cut value in this case is $c^*_1=4n+2\frac nt$.

  Observing that $c^*_0/c^*_1=\frac{4n}{4n+2n/t} = \frac {2t}{2t+1} <
  1-\eps$, we conclude that a randomized one-way protocol for
  $\MC^\eps$ (on input size $n'=4n+n/t=O(n)$) gives a randomized
  one-way protocol for $\BHH_n^t$. By Lemma \ref{lem:bhh_lb} the
  Theorem follows.
\end{proof}

\definecolor{myblue}{RGB}{162,213,247}
\definecolor{mygreen}{RGB}{157,252,143}
\definecolor{myred}{RGB}{252,143,175}

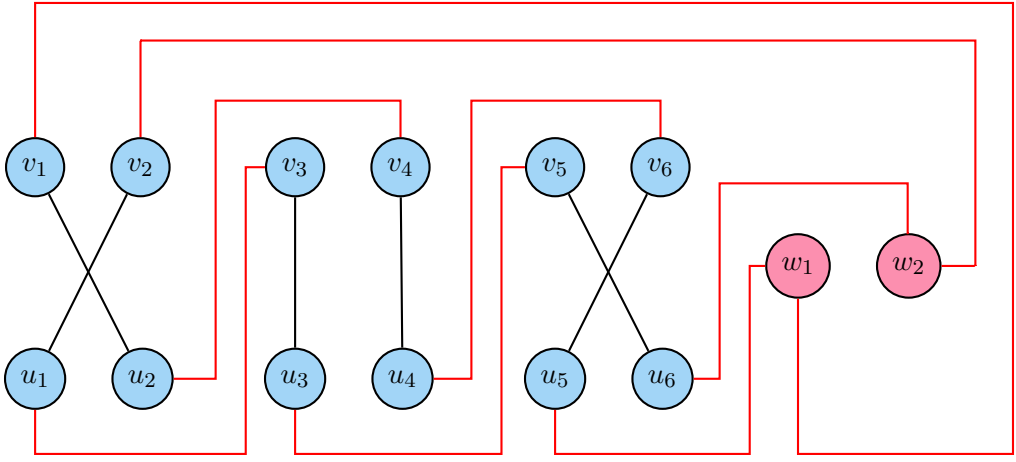
\begin{figure}
\begin{center}

\iffigfromfile
    \includegraphics[scale=1]{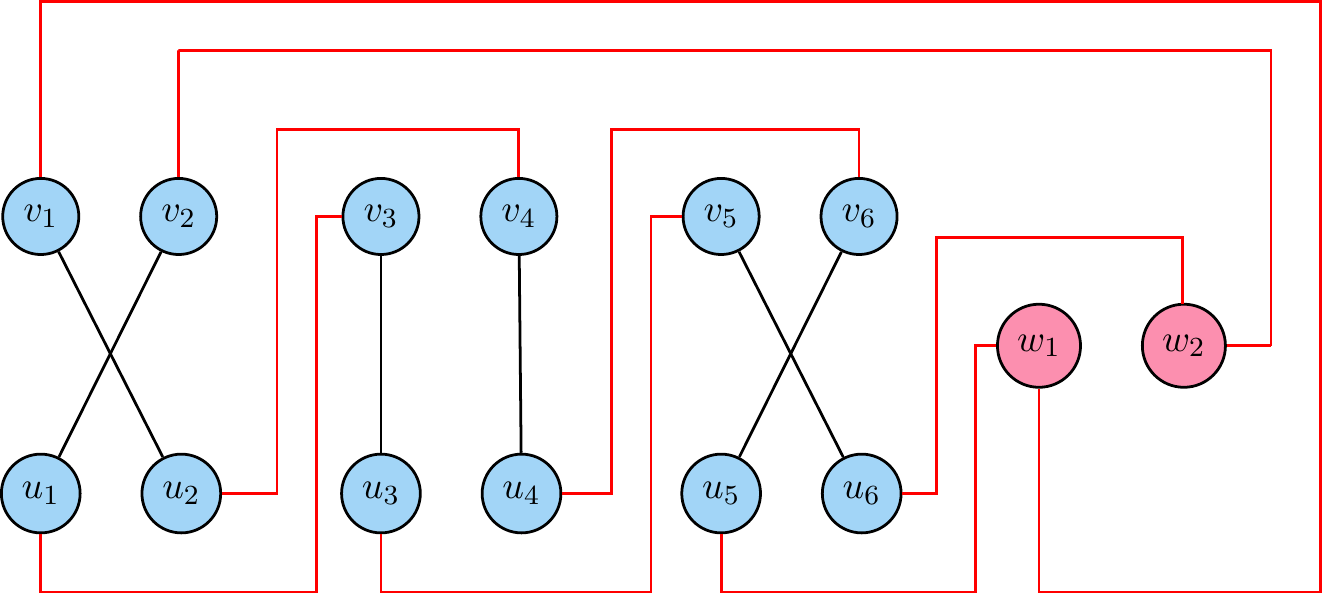}
\else %

\tikzset{external/force remake}
\tikzsetnextfilename{figure-gadget}
\begin{tikzpicture}[thick,
  every node/.style={draw,circle},
  unode/.style={fill=myblue},
  vnode/.style={fill=mygreen},
  wnode/.style={fill=myred},
  every fit/.style={ellipse,draw,text width=2cm},
  bend angle=45
]
\tikzset{external/force remake}
\node [unode] (u1) {$u_1$};
\node [unode,right=0.6cm of u1] (u2) {$u_2$};
\node [unode,above=2cm of u1] (v1) {$v_1$};
\node [unode,right=0.6cm of v1] (v2) {$v_2$};
\node [unode,right=1.2cm of u2] (u3) {$u_3$};
\node [unode,above=2cm of u3] (v3) {$v_3$};
\node [unode,right=0.6cm of u3] (u4) {$u_4$};
\node [unode,right=0.6cm of v3] (v4) {$v_4$};
\node [unode,right=1.2cm of u4] (u5) {$u_5$};
\node [unode,above=2cm of u5] (v5) {$v_5$};
\node [unode,right=0.6cm of u5] (u6) {$u_6$};
\node [unode,right=0.6cm of v5] (v6) {$v_6$};

\node [wnode,above right=0.9cm and 1.2cm of u6] (w1) {$w_1$};
\node [wnode,right=0.6cm of w1] (w2) {$w_2$};

\draw (u1) -- (v2);
\draw (u2) -- (v1);
\draw (u3) -- (v3);
\draw (u4) -- (v4);
\draw (u5) -- (v6);
\draw (u6) -- (v5);

\draw[color=red] (u1.south) |- (2.8,-1) |-(v3.west);
\draw[color=red] (u2.east) -| (2.4,0) -|  (2.4,3.7) -|
(4.85,3.7) -| (v4.north);

\draw[color=red] (u3.south) |- (6.2,-1) |-(v5.west);
\draw[color=red] (u4.east) -| (5.8,0) -|  (5.8,3.7) -|
(8.3,3.7) -| (v6.north);

\draw[color=red] (u5.south) |- (9.5,-1)
|-(w1.west);

\draw[color=red] (u6.east) -| (9.1,0) -| (9.1,2.6) -| (11.6,2.6)
|-(w2.north);

\draw[color=red] (w1.south) |- (13,-1) |- (13,5) |- (0,5) |-
(v1.north);
\draw[color=red] (w2.east) -| (12.5,1.5) -| (12.5,4.5) -|
(1.4,4.5) -| (v2.north);
\end{tikzpicture}

\fi %

  \caption{An example of a gadget constructed in the proof of Theorem
    \ref{thm:oneway_lb} for $t=3$, a matching $M$ that contains the
    hyperedge $M_1=(1,2,3)$, $x_1=1$, $x_2=0$, $x_3=1$ and
    $w_1=0$. The result is two paths of length $7$. Alice's and Bob's
    edges are colored black and red respectively.
    \%vspace{20pt}
   }
\hrulefill
\end{center}
\label{fig:bhh_mc_red}
\end{figure}

\begin{proof}[Proof of Theorem~\ref{thm:oneway_lb}]
  Any streaming algorithm for $\MC^\eps$ leads to a one-way
  communication protocol in the two party setting. Moreover the
  communication complexity of this protocol is exactly the space
  complexity of the streaming algorithm. Hence by Theorem
  \ref{thm:oneway_lb} the streaming space complexity is at least as
  high as the one way randomized communication complexity.
\end{proof}

\subsection{Reporting a Vertex-Bipartition (rather than a value)}
\label{lb_actual_cut}
\newcommand{\sk}{\mathbf{sk}} \newcommand{\est}{\mathbf{est}}

We show a simple $\Omega(n)$ space lower bound for reporting a
vertex-bipartition that gives an approximate maximum cut.

\begin{proposition}
  Let $\eps\in(0,\frac 12)$ be some small constant. Suppose $\sk$ is a
  sketching algorithm that outputs at most $s=s(n,\eps)$ bits, and
  $\est$ is an estimation algorithm, such that together, for every
  $n$-vertex graph $G$, (with high probability) they output a
  vertex-bipartition that gives an approximately maximum cut; i.e.,
  $\est(\sk(G))=S$ such that $w(S,\bar S)\ge(1-\eps)\tilde w$ where
  $\tilde w$ is the maximum cut in $G$. Then $s\ge\Omega_\eps(n)$.
\end{proposition}
\begin{proof}
  \newcommand {\code}{\mathcal C}
  Let $\code\subset\{0,1\}^n$ be a binary error-correcting code of
  size $|\code|=2^{\Omega(n)}$ with relative distance $\eps$. We may
  assume w.l.o.g. that for every $x\in \code$ the hamming weight $|x|$
  is exactly $n/2$ (for instance by taking $\code'=\{x\bar x: x\in
  \code\}$ where $\bar x$ denotes the bitwise negation of $x$), and
  that there are no $x,y\in \code$ such that $|x-\bar y|\le \frac \eps
  2 n$ (since for every $x\in\code$ there could be at most one ``bad''
  $y$, and we can discard one codeword out of every such pair).

  Fix a codeword $x\in\{0,1\}^n$ and consider the complete bipartite
  graph $G_x=(V,E)$ where $V=[n]$ and $E=\{(i,j):x_i=0\wedge x_j=1\}$.
  The maximum cut value in $G_x$ is obviously $\tilde w = n^2/4$.  Let
  $y\in\{0,1\}^n$ such that $\frac 12\eps n \le |x-y| \le \frac n2$.
  Identifying $x,y$ with subsets $S_x,S_y\subseteq[n]$, and using
  the fact that $|S_x\triangle S_y| = |x-y| \ge\frac12\eps n$, the value of
  the cut $(S_y,\bar S_y)$ in $G_x$ is
  \[
  |E(S_y,\bar S_y)|=\frac {n^2}4- |S_x\setminus S_y| \left(\frac n2 -
    |S_y\setminus S_x|\right) - |S_y\setminus S_x| \left(\frac n2 -
    |S_x\setminus S_y|\right)< (1-\Omega(\eps))\frac {n^2}4
  \]

  Let $\sk(G_x)$ be the sketch of $G_x$, and let $\est(\sk(G_x))=S$ be
  the output of the estimation algorithm on the sketch of
  $G_x$. Therefore if the sketch succeeds (which by our assumption
  happens with high probability) and the cut $(S,\bar S)$ has value at
  least $(1-\Omega(\eps))\tilde w$, then by the preceding argument the
  corresponding vector $x_S$ is of relative hamming distance smaller
  than $\frac\eps2$ from $x$ and then one can decode $x$ from
  $S$.\footnote{Since the cuts $(\bar S,S)$ has the same value as
    $(S,\bar S)$, the vector $x_S$ can actually be $\eps$-close to
    $\bar x$, but by taking our code to have no codeword being close
    to the negation of another codeword we can always try decoding
    both $x_S$ and $\bar x_S$.} By standard arguments from information
  theory, the size $s$ of a sketch that succeeds with high probability
  must be at least $\Omega(\log |C|)=\Omega_\eps(n)$.
\end{proof}

\subsection{$2/3$-Approximation of \MC in the Two Party Model}

We remark that in the one-way two-party model, the parameter range
$\eps\in(0,\frac 15)$ in Theorem \ref{thm:oneway_lb} is tight and not
merely a technical limitation of our analysis. In that model, the
problem of giving a $\frac {2t}{2t+1}$-approximation of the maximum
cut exhibits an exponential gap in the communication complexity
between the case of $t\ge 2$, where we have shown that a polynomial
number of bits is necessary, and the case $t=1$, for which
logarithmically many bits suffice, as follows from the following
simple protocol.

\begin{proposition}
  Let $G=(V,E_A\cupdot E_B)$ be an input graph on $|V|=n$ vertices. Let
  $w_A$ and $w_B$ be the maximum cut values in $G_A=(V,E_A)$ and
  $G_B=(V,E_B)$ respectively. Then it holds for the maximum cut value
  $w$ in $G$
  \[
  \tfrac 23 (w_A+w_B) \le w \le w_A+w_B
  \]
\end{proposition}
\begin{proof}
  Consider cuts $C_A,C_B:V\rightarrow \{0,1\}$ such that $w(C_A)=w_A$
  and $w(C_B)=w_B$. Let ${C:V\rightarrow \{0,1\}}$ be a cut chosen
  uniformly at random from $\{C_A,C_B,C_A+C_B\}$. For an edge
  $e=(u,v)\in C_A$, either $C_B(u)+C_B(v)=1$ or
  $(C_A+C_B)(u)+(C_A+C_B)(u) = \left( C_A(u)+C_A(v) \right) + \left(
    C_B(u)+C_B(v) \right) = 1 + 0 = 1$. Either way
  $\Pr_{C\in_R\{C_A,C_B,C_A+C_B\}}[e\in C] = \frac 23$. Similarly the
  same holds for an edge $e\in C_B$. Therefore by linearity of
  expectation a random cut in $\{C_A,C_B,C_A+C_B\}$ has value at least
  $\frac 23 (w_A+w_B)$.  The second inequality is trivial.
\end{proof}

\begin{corollary}
  The one-way communication complexity of $\MC^{1/3}$ is $O(\log n)$.
\end{corollary}
\begin{proof}
  Alice computes the value $w_A$ and sends it to Bob. Bob computes the
  value $w_B$ and outputs $\frac 23(w_A+w_B)$.
\end{proof}

\section{Sketching Cuts in Hypergraphs}
\label{sec:hypergraphs}

In a celebrated series of works,
Karger~\cite{karger1995random,karger1998better,karger1999random}
and Bencz\'ur and Karger~\cite{benczurkarger1996,BK02} showed
an effective method to sketch the values of \emph{all} the cuts of
an undirected (weighted) graph $G=(V,E,w)$
by constructing a \emph{cut-sparsifier},
which is a subgraph with different edge weights,
that contains only $\tilde O\left(n/\eps^2 \right)$ edges,
and approximates the weight of every cut
in $G$ up to a multiplicative factor of $(1\pm\eps)$.
We generalize the ideas of Bencz\'ur and Karger to obtain cut-sparsifiers of hypergraphs,
as stated below.
Such sparsifiers (and sketches) can be computed by streaming algorithms
that use $\tilde O(rn)$ space for $r$-uniform hypergraphs
using known techniques (of \cite{ag09} and subsequent work).

\begin{theorem} \label{thm:hypersparsifiers}
  For every $r$-uniform%
\footnote{Throughout this work we allow $r$-uniform
    hypergraphs to contain also hyperedges with less than $r$ endpoints
    (for instance by allowing duplicate vertices in the same hyperedge).
}
hypergraph $H=(V,E,w)$ and an error parameter $\eps\in(0,1)$,
there is a subhypergraph $H_\eps$ (with different edge weights) such that:
    \begin{itemize} \compactify
    \item $H_\eps$ has $O\left(n(r+\log n) / \eps^{2} \right)$
      hyperedges.
    \item The weight of every cut in $H_\eps$ is $(1\pm\eps)$ times the
      weight of the corresponding cut in $H$.
    \end{itemize}
\end{theorem}

A key combinatorial property exploited in the Bencz\'ur-Karger
analysis is an upper bound on the number of cuts of near-minimum
weight~\cite{karger1993global}.  It asserts that in the number of
minimum-weight cuts in an $n$-vertex graph is at most $n^2$ (which had
been previously shown by \cite{lomonosov1972lower} and
\cite{dinitz1976structure}), and more generally, there are at most
$n^{2\alpha}$ cuts whose weight is at most $\alpha\ge1$ times the
minimum.  Correctly generalizing this property to $r$-uniform
hypergraphs appears to be a nontrivial question.  A fairly simple
analysis generalizes the latter bound to $n^{r\alpha}$, but using new
ideas, we manage to obtain the following tighter bound.
\begin{theorem}\label{thm:numcuts}
  Let $H=(V,E,w)$ be a weighted $r$-uniform hypergraph with $n$
  vertices and minimum cut value $\hat w$. Then for every half-integer
  $\alpha\ge 1$, the number of cuts in $H$ of weight at most $\alpha
  \hat w$ is at most $O(2^{\alpha r}n^{2\alpha})$.
\end{theorem}
We prove this ``cut-counting'' bound in Section~\ref{sec:NearMinH}.
With this bound at hand, we prove Theorem~\ref{thm:hypersparsifiers}
similarly to the original proof of~\cite{benczurkarger1996} for
graphs, as outlined in Subsection~\ref{apx:fullkarger}.

\medskip Cuts in hypergraphs are perhaps one of the simplest examples
of CSPs, with the hypergraph vertices becoming boolean variables, and
the hyperedges becoming constraints defined by the predicate
\textsc{Not-All-Equal}. A natural question is whether general CSPs
admit sketches as well, where a sketch should provide an approximation
to the value of every assignment to the CSP (as usual, the value of an
assignment is the number of constraints it satisfies).  Although we
are still far from answering this question in full generality, we
prove that for the well-known SAT problem, sketching is indeed
possible.
\begin{theorem}
  \label{thm:sat}
  For every error parameter $\eps\in(0,1)$, there is a sketching
  algorithm that produces from an $r$-\textsc{CNF} formula $\Phi$ on
  $n$ variables a sketch of size $\tilde O(rn/\eps^2)$, that can be
  used to $(1\pm\eps)$-approximate the value of every assignment to
  $\Phi$.
\end{theorem}

\subsection{Near-Minimum Cuts in Hypergraphs}
\label{sec:NearMinH}

In this subsection we prove our upper bound on the number of
near-minimum cuts (Theorem \ref{thm:numcuts}).  We generalize Karger's
min-cut algorithm to hypergraphs, and then show that its probability
to output any individual cut is not small (Theorem~\ref{thm:probcut}),
which immediately yields a bound on the number of distinct cuts.
Finally, we show that the exponential dependence on $r$ in
Theorem~\ref{thm:numcuts} is necessary (Section~\ref{sec:CutsLB}).

\subsubsection{A Randomized Contraction Algorithm}
Consider the following generalization of Karger's contraction
algorithm \cite{karger1993global} to hypergraphs.

\begin{algorithm}
  \caption{\textsc{ContractHypergraph}}
  \label{alg:contract}
\begin{algorithmic}[1]
  \Require
  \Statex an $r$-uniform weighted hypergraph $H=(V,E,w)$
  \Statex a parameter $\alpha>1$
  \Ensure a cut $C=(S,V\setminus S)$
  \State $H'\gets H$
  \While{$|V(H')| > \alpha r$ }
  \State{$e\gets$ random hyperedge in $H'$ with probability proportional to its
    weight} \State{contract $e$ by merging all its endpoints and
    removing self-loops\footnotemark}
  \EndWhile
  \State{$C'\gets $ random cut in $H'$
(bipartition of $V(H')$)} \label{alg:finalstep}
  \State{\textbf{return} the cut $C$ in $H$ induced by the cut $C'$}
\end{algorithmic}
\end{algorithm}
\footnotetext{Here by self-loops we refer to hyperedges that contain
  only a single vertex. Note also that the cardinality of an edge can
  only decrease as a result of contractions.}
\begin{theorem} \label{thm:probcut} Let $H=(V,E,w)$ be a weighted
  $r$-uniform hypergraph with minimum cut value $\hat w$, let $n=|V|$,
  and let $\alpha\ge 1$ be some half-integer. Fix $C=(S,V\setminus S)$
  to be some cut in $H$ of weight at most $\alpha \hat w$. Then
  Algorithm \ref{alg:contract} outputs the cut $C$ with probability at
  least $\frac {Q_{n,r,\alpha}}{2^{\alpha r-1}-1}$ for
  \[
  Q_{n,r,\alpha}= \frac {2\alpha+1}{(r+1)}{{n-\alpha(r-2)} \choose
    2\alpha}^{-1}
  \]
\end{theorem}

Since Theorem \ref{thm:probcut} gives a lower bound on the probability
to output a specific cut (of certain weight), and different cuts
correspond to disjoint events, it immediately implies an upper bound
on the possible number of cuts of that weight, proving Theorem
\ref{thm:numcuts}.

\begin{proof}
  Fix $C=(S,V\setminus S)$ to be some cut of weight $\alpha \hat w$ in
  $H$. For $t=1,\dots,n$, denote by $I_t$ the iteration of the
  algorithm where $H'$ contains $t$ vertices. Since a contraction of a
  hyperedge may reduce the number of vertices by anywhere between $1$
  and $r-1$, in a specific execution of the algorithm, not necessarily
  all the $\{I_t\}_{t=1}^n$ occur. Similarly, let the random variable
  $E_t$ be the edge contracted in iteration $t$.

  We say that an iteration $I_t$ is \emph{bad} if $E_t\in C$ (i.e., the
  hyperedge contains vertices from both $S$ and $V\setminus
  S$). Otherwise, we say it is \emph{good} (including iterations that
  do not occur in the specific execution such as $I_1,\dots,I_{\alpha
    r}$). For any fixed $e_n,\dots,e_{t+1}\in E$ define
  \begin{align*}
    q_t(e_n,\dots,e_{t+1}) &= \Pr\left[ I_t,I_{t-1},\dots,I_1 \text{ are
        good}|E_n=e_n,\dots,E_{t+1}=e_{t+1}\right]
  \end{align*}
  Note that $q_n$ is simply the probability that all iterations of the
  algorithm are good i.e., no edge of the cut $C$ is contracted. When
  that happens, in step \ref{alg:finalstep} of the algorithm, there
  exists a cut $C'$ in $H'$ that corresponds to the cut $C$ in
  $H$. Since at that stage, there are at most $\alpha r$ vertices in
  $H'$, the probability of choosing $C'$ is at least $\frac
  1{2^{\alpha r-1}-1}$. Hence the overall probability of outputting
  cut $C$ is at least $q_n\cdot\frac 1{2^{\alpha r-1}-1}$. We thus
  need to give a lower bound on $q_n$. To this end we
  prove the following lemma.
  \begin{lemma}\label{lem:bound_qn}
    $q_t(e_n,\dots,e_{t+1}) \ge Q_{t,r,\alpha}$ for every $t=\alpha
    r,\dots,n$, and every $e_n,\dots,e_{t+1}\in E\setminus C$.
  \end{lemma}
  Using the lemma for $t=n$ bounds the overall probability of
  outputting cut $C$ and proves Theorem \ref{thm:probcut}.
\end{proof}

\subsubsection{Proof of Lemma \ref{lem:bound_qn}}
By (complete) induction on $t$. For the base case, note that
$q_t(e_n,\dots,e_{t+1})=1$ for $t=1,\dots,\alpha r$ since no
contractions take place in those iterations.

For the general case, fix an iteration $I_t$ and from now on,
condition on some set of values ${E_n=e_n,\dots,E_{t+1}=e_{t+1}}$. All
probabilities henceforth are thus conditioned, and for brevity we omit
it from our notation. Observe that depending on the cardinality of
$E_t$, the next iteration (\emph{after} iteration $I_t$) may be one of
$I_{t-1},\dots,I_{t-r+1}$. Let $p_i=\Pr[|E_t|=i]$ and let
$y_i=\Pr\left[|E_t|=i\wedge E_t\in C\right]$.\footnote{Since not all
  iterations occur in all executions, it might be the case that
  \emph{no} edge is contracted in iteration $t$. However, in that case
  iteration $t$ is good, and hence by the induction hypothesis the
  claim holds.}$^{,}$\footnote{Note that $|e|$ refers to the edge's
  cardinality, whereas $w(e_i)$ refers to its weight.}  We can now
write a recurrence relation:
\begin{align*}
  q_t(e_n,\dots,e_{t+1})&=\Pr[I_t,\dots,I_1\text{ are good}| E_n=e_n,\dots,E_{t+1}=e_{t+1}] \\
  &= \sum_{i=2}^r \Pr\left[|E_t|=i\wedge E_t\notin C\right] \cdot
  \Pr[I_{t-i+1},\dots,I_1\text{ are good}||E_t|=i, E_t \notin C]  \\
  &= \sum_{i=2}^r (p_i-y_i)
  \E_{E_t}\left[q_{t-i+1}(e_n,\dots,e_{t+1},E_t)| |E_t|=i\wedge E_t\notin C\right]\\
  &\ge \sum_{i=2}^r (p_i-y_i)Q_{t-i+1,r,\alpha}
\end{align*}

For $i=2,\dots,r$ let $W_i=\sum_{e'\in H':|e'|=i} w(e')$ be the total
weight of hyperedges in $H'$ of cardinality $i$ (at iteration $t$) and
let $W=\sum_{i=2}^rW_i$ be the total weight in $H'$.

Observe that $p_i=\frac {W_i}W$ since $E_t$ is chosen with probability
proportional to the hyperedge's weight, and $\sum_{v\in V'}deg(v) =
\sum_{i=2}^ri\cdot W_i$ since a hyperedge of cardinality $i$ is
counted $i$ times on the lefthand side. By averaging, there exists a
vertex $v\in V(H')$ such that $deg(v)\le \frac 1t \sum_{i=2}^ri\cdot
W_i$, and since it induces a cut in $H$ whose weight is
exactly $deg(v)$, we obtain that $\hat w \le deg(v) \le \frac 1t
\sum_{i=2}^ri\cdot W_i$.

Next note that
\[
\sum_{i=2}^ry_i=\Pr[E_t\in C] \le \frac {\alpha \hat w}W \le
\tfrac \alpha t \sum_{i=2}^ri\cdot \frac {W_i}W = \tfrac \alpha t
\sum_{i=2}^ri\cdot p_i
\]
where the first inequality uses the conditioning on all previous
iterations being good, which means that all hyperedges in $C$
have survived in $H'$, and thus $w_H(C)=w_{H'}(C)$.

Altogether, to prove the lemma it suffices to show that the value of
the following linear program is at least $Q_{t,r,\alpha}$ (and from
now on we omit the subscripts $r$ and $\alpha$, denoting $Q_t=Q_{t,r,\alpha}$).

\begin{alignat*}{3}
  \text{minimize } &\sum_{i=2}^r(p_i-y_i)Q_{t-i+1} \\
  \text{subject to } &
  0\le y_i \le p_i \qquad\qquad \forall i=2,\dots,r\\
  &\sum_{i=2}^rp_i = 1  \\
  &\sum_{i=2}^r y_i \le \tfrac \alpha t \sum_{i=2}^ri\cdot p_i .
\end{alignat*}
First observe that the last constraint implies
\begin{equation} \label{eq:psmallerthany}
\sum_{i=2}^r y_i \le \tfrac \alpha t \sum_{i=2}^ri\cdot
p_i \le \tfrac \alpha t \sum_{i=2}^rr\cdot
p_i = \tfrac {\alpha r}t \sum_{i=2}^rp_i < \sum_{i=2}^rp_i ,
\end{equation}
which means that in every \emph{feasible} solution there is always
some $y_i<p_i$. This implies that in every \emph{optimal} solution,
the last constraint is tight, since otherwise increasing such a $y_i$
will decrease the value of the solution, without violating any of the
other constraints.

It is easy to see that this linear program is both feasible and
bounded, and therefore has an optimal solution that is basic (i.e., a
vertex of the polytope). The dimension of the linear program (i.e.,
the number of variables) is $2r-2$, and thus in a basic feasible
solution (at least) $2r-2$ of the $2r$ constrains must be
tight. Therefore there are at most $2$ untight constraints among the
$2r-2$ constraints $0\le y_i \le p_i$, meaning there are at most $2$
indices $i,j$ such that $p_i\neq 0$.  We proceed by analyzing the
possible cases:
\begin{itemize}
\item $0<y_i=p_i$ and $0<y_j=p_j$. This case is not possible, since
  that would have implied $\sum_{i=2}^r y_i = \sum_{i=2}^rp_i$,
  contradicting \eqref{eq:psmallerthany}.
\item $0=y_i<p_i$ and $0=y_j<p_j$. This case is also not possible
  since that would have implied $\sum_{i=2}^r y_i = 0$, contradicting
  the tightness of the last constraint in an optimal solution.
\item $0=y_i<p_i$ and $0<y_j=p_j$. Since all other $p_\ell=0$, the
  other LP constraints become
  \begin{gather*}
    p_i+p_j=1\\
    0+p_j=y_i+y_j=\tfrac \alpha t(ip_i+jp_j)
  \end{gather*}
  Solving the two equations we obtain:
  \begin{equation} \label{eq:case3} \textsc{LP}=\left(1-\tfrac
      {\alpha i}{t+\alpha i-\alpha j} \right)Q_{t-i+1} \ge
    \left(1-\tfrac {\alpha i}{t+\alpha i-\alpha r} \right)Q_{t-i+1} =
    \tfrac {t-\alpha r}{t+\alpha i -\alpha r}Q_{t-i+1}
  \end{equation}
  To use the induction hypothesis, we distinguish between two cases:
  \begin{enumerate}
  \item $t-i+1\ge \alpha r$, in which case it is thus sufficient to
    prove the following claim.
    \begin{claim}
      For every half-integer $\alpha\ge 1$ and integers $r \ge i \ge
      2$ and $t\ge \alpha r + i- 1$, it holds
      $\frac{Q_{t-i+1,r,\alpha}}{Q_{t,r,\alpha}} \ge \frac {t+\alpha i
        -\alpha r}{t-\alpha r}$.
    \end{claim}
    \begin{proof}
      Recall that  $Q_t= \frac {2\alpha+1}{(r+1)}{{t-\alpha(r-2)} \choose
        2\alpha}^{-1}$ and denote $t'=t-\alpha r$. Then
      \begin{align*}
        \textsc{LHS} &= \frac { {{t'+2\alpha} \choose 2\alpha} } {
          {{t'-i+2\alpha+1} \choose 2\alpha} } = \frac
        {(t'+2\alpha)\cdots(t'+1)}{(t'+2\alpha-i+1)\cdots(t'+1-i+1)} =
        \frac
        {(t'+2\alpha)\cdots(t'+2\alpha-i+2)}{t'\cdots(t'-i+2)} \\
        &= \left(1+\frac {2\alpha}{t'} \right)\cdots\left(1+\frac
          {2\alpha}{t'-i+2} \right) \ge \left(1+\frac {2\alpha}{t'}
        \right)^{i-1} \ge 1+\frac {2\alpha(i-1)}{t'} \ge 1+\frac{\alpha i
        }{t'} = \textsc{RHS}
      \end{align*}
    \end{proof}
  \item $t-i+1 < \alpha r$, in which case $Q_{t-i+1}=1$. Here we get
    \[
    \textsc{LP} \ge 1-\tfrac {\alpha i}{t-\alpha r+ \alpha i} \ge 1 - \tfrac
    {\alpha i}{\alpha i+1} = \tfrac 1{\alpha i+1} \ge \tfrac 1{\alpha
      r+1} \ge \tfrac {2\alpha +1}{(r+1){{t-\alpha(r-2)} \choose
        2\alpha }} = Q_t ,
    \]
    where the last inequality follows from the fact that $t-\alpha(r-2)\ge
    \alpha r +1 -\alpha(r -2) \ge 2\alpha + 1$.
  \end{enumerate}
\item $0<y_i<p_i$ and $0=y_j=p_j$. In this case $p_i=1$, $y_i=\frac
  {\alpha i}t$, and therefore
  \[
  \textsc{LP}=\left(1-\tfrac {\alpha i} t\right)Q_{t-i+1} \ge
  \left(1-\tfrac {\alpha i}{t-\alpha(r-i)}\right)Q_{t-i+1} ,
  \]
  which is exactly as in \eqref{eq:case3} in the previous case.
\end{itemize}
Having bounded the value of the linear program, this completes the
proof of Lemma \ref{lem:bound_qn}.

\subsubsection{Lower Bound}
\label{sec:CutsLB}

For completeness, we remark that at least for $\alpha > 1$, the
exponential dependence on $r$ in Theorem \ref{thm:numcuts} is
indeed necessary. Consider a ``sunflower'' hypergraph on $n=rm-m+1$
vertices that consists of $m$ hyperedges of size $r$, intersecting at
a single vertex, supplemented with $m$ two-uniform cliques of size $r$
each -- one for each of the hyperedges. Each of the $r$-hyperedges is
given weight $1$ and each of the two-edges is given weight $\frac
{\alpha-1}{2^r}$.  The minimum cut value in this graph is $1$, since
every cut contains at least one of the $r$-hyperedges. However, all
$\Omega(m\cdot 2^r)$ cuts given by the $2^r$ bipartitions of a single
$r$-hyperedge, are of weight at most $\alpha$.

\subsection{Proof Of Theorem \ref{thm:hypersparsifiers}}
\label{apx:fullkarger}
Since the proof of Theorem \ref{thm:hypersparsifiers} closely follows
the proof in the original setting of graphs (cf.~\cite{BK02}), we
refrain from repeating the full details. Instead, we choose to present
an outline of the proof, emphasizing the key reasons it translates to
the hypergraph setting, and handling the key differences, that require
a separate treatment.

The main tool used by Bencz\'ur and Karger is random sampling: each
edge $e$ is included in the sparsifier with probability $p_e$, and
given weight $w_e/p_e$ if included. It is thus immediate that every cut
in the sparsifier preserves its weight in expectation. The main task is
thus to carefully select the sampling probabilities $p_e$ in order
to both obtain the required number of edges in the sparsifier, and
guarantee the required concentration bounds.

As a rough sketch, to guarantee concentration, one needs to apply a
Chernoff bound to estimate the probability that a specific cut (which
is a sum of the independent samples of the edges it contains) deviates
from its expectation. Subsequently, a union bound over all cuts is
used to show the concentration of \emph{all} cuts. Yet a-priori it is
unclear whether the Chernoff bound is strong enough to handle the
exponentially many different cuts in the union bound. The remedy comes
in the form of the bound on the the number of cuts of each weight
given by Theorem \ref{thm:numcuts}. It is still unclear how should the
random sampling be tuned to handle both the small and large cuts
simultaneously. If we are to chose the sampling probability to be
small enough to handle the exponentially many large cuts, we run into
trouble of small cuts having large variance. On the other hand,
increasing the sampling probability imposes a risk of ending up with
too many edges in the sparsifier.

Following Bencz\'ur and Karger, we now show that when no edge carries a
large portion of the weight in any of the cuts, the cut-counting theorem is
sufficient to obtain concentration.

\begin{theorem} \label{thm:nosmallcuts} Let $H=(V,E,w)$ be a
  $r$-uniform hypergraph on $n$ vertices, let $\eps>0$ be an error
  parameter, and fix $d\ge1$. If $H'=(V,E',w')$ is a random
  subhypergraph of $H$ where the weights $w'$ are independent random
  variables distributed arbitrarily (and not necessarily identically)
  in the interval $[0,1]$, and the expected weight of every cut in
  $H'$ exceeds $\rho_\eps=\frac 3{\eps^2}\left(r+(d+2)\ln n\right)$,
  then with probability at least $1-n^{-d}$, every cut in $H'$ has
  weight within $(1\pm\eps)$ of its expectation.
\end{theorem}

One can verify that the proof of a similar theorem for the case of
graphs, as appears in \cite{karger1999random}, translates to the
hypergraph setting. For the sake of the proof, a cut is merely a sum
of independently sampled edges/hyperedges. The lower bound on the
weight of the minimum expected cut $\hat w$ allows one to show that
probability of a cut of weight $\alpha\hat w$ to deviate from its
expectation is at most $n^{-\alpha(d+2)}\cdot e^{-\alpha r}$ which trades-off nicely with
the bound on the number of cuts given by Theorem \ref{thm:numcuts}.

Informally, the latter theorem implies that in order to obtain the
desired concentration bound in the general case, the sampling
probability of an edge must be inversely proportional to the size of
the largest cut that contains that edge. This motivates the following
definitions, and the theorem that follows them.

\begin{definition}
  A hypergraph $H$ is $k$-connected if the weight of each cut in $H$
  is at least $k$.
\end{definition}

\begin{definition}
  A $k$-strong component of $H$ is a maximal $k$-connected
  vertex-induced subhypergraph of $H$.
\end{definition}

\begin{definition}
  The strong connectivity of hyperedge $e$, denoted $k_e$, is the
  maximum value of $k$ such that a $k$-strong component contains (all
  endpoints of) $e$.
\end{definition}

Note that one can compute the strong connectivities of all hyperedges
in a hypergraph in polynomial time as follows. Compute the global
minimum cut, and then proceed recursively into each of the two
subhypergraphs induced by the minimum cut. The strong connectivity of an
edge would then be the maximum among the minimum cuts of all the
subhypergraphs it has been a part of throughout the recursion. The minimum
cut in a hypergraph was shown to be computable in $O(n^2\log n+mn)$
time by \cite{klimmekwagner}. Note that since the total number of
subhypergraphs considered throughout the recursion is at most $n$, there
are at most $n$ different strong-connectivity values in any
hypergraph.

\begin{theorem}
  Let $H$ be an $r$-uniform hypergraph, and let $\eps>0$ be an error
  parameter. Consider the hypergraph $H_\eps$ obtained by sampling
  each hyperedge $e$ in $H$ independently with probability $p_e =
  \frac {3((d+2)\ln n + r)}{k_e\eps^2}$, giving it weight $1/p_e$ if
  included. Then with probability at least $1-O(n^{-d})$
\begin{enumerate}\compactify
\item The hypergraph $H_\eps$ has $O\left(\frac n{\eps^2}(r+\log
    n)\right)$ edges.
\item Every cut in $H_\eps$ has weight between $(1-\eps)$ and
  $(1+\eps)$  times its weight in $H$.
\end{enumerate}
\end{theorem}

The proof of the theorem for the hypergraph setting is again
completely identical to the proof in \cite{BK02} (c.f., Theorem
2.6). The only thing that needs verifying is that strong-connectivity
induces a recursive partitioning of the vertices of the hypergraph,
just as it does when dealing with graphs. This is in fact the case,
mainly because the components considered in the definitions are
vertex-induced, and therefore the cardinality of the hyperedges plays
no part. One can then decompose the hyperedges of the hypergraph to
``layers'', based on their strong-connectivity, and apply Theorem
\ref{thm:nosmallcuts} to each layer separately.

To complete our discussion we bring the reader's attention to a couple
of places where the cardinality of the hyperedges has played part:
\begin{itemize}\compactify
\item The modified parameter $p_e= \frac {3(r+(d+3)\ln n)}
  {k_e\eps^2}$ counters the number of cuts from Theorem
  \ref{thm:numcuts} (at most $O(2^{\alpha r}n^{2\alpha})$ cuts of
  weight $\alpha \hat w$) and the number of distinct edge-connectivity
  values, which is at most $n$.\footnote{In their analysis \cite{BK02}
    take a union bound over $n^2$ distinct edge-connectivity
    values. For hypergraphs using the stronger linear bound (instead of the
    trivial $n^r$) is crucial.}
\item The number of edges in the sparsifier is (with high probability)
  $O\left(\frac{n}{\eps^2}(r + \log n)\right)$ since the sampling
  probability is also linear in $r$.
\end{itemize}

\subsection{SAT Sparsification}
\label{sec:SAT}

\begin{lemma} \label{lem:sat2hg} Given an $r$-CNF formula $\Phi$ with
  $n$ variables and $m$ clauses, there exists an {$(r+1)$-uniform}
  hypergraph $H$ with $2n+1$ vertices, and a mapping
  $\Pi:\{0,1\}^n\rightarrow \{0,1\}^{2n+1}$ from the set of all
  assignments to $\Phi$ to the set of all cuts in $H$, such that for
  every assignment $\phi$, it holds that
  $val_\Phi(\phi)=val_H(\Pi(\Phi))$.
\end{lemma}

\begin{proof}
  Consider an $r$-CNF formula $\Phi$ with variables $\{x_i\}_{i\in
    [n]}$. We construct the weighted hypergraph $H$ whose vertices are
  $\{x_i,\neg x_i\}_{i\in [n]}$ and a special vertex $F$. For each
  clause $\ell_{i_1}\vee\ell_{i_2}\vee\cdots\vee\ell_{i_r}$, we add a
  hyperedge $\{\ell_{i_1},\ell_{i_2},\dots,\ell_{i_r},F\}$. Moreover,
  let $\Pi$ be the mapping that maps an assignment to $\Phi$ to the
  cut in $H$ obtained by placing all vertices corresponding to true
  literals on one side, and the $F$ vertex together with all vertices
  corresponding to false literals on the other side.

  For an assignment $\phi$ to $\Phi$, it is clear that a hyperedge is
  contained in the cut $\Pi(\phi)$ if and only if at least one of the
  vertices it contains is on the opposite side of $F$. Therefore the
  weight of $\Phi(\phi)$ is exactly the value of $\phi$.
\end{proof}

Theorem \ref{thm:sat} follows from Lemma \ref{lem:sat2hg} and
Theorem \ref{thm:hypersparsifiers}.

\section{Future Directions}
\label{sec:future}

Our results raise several questions that deserve further work.

\begin{description}

\item[Sketching \MC.] Our results and the results of \cite{KKS14} make
  progress on the \emph{streaming} complexity of approximating \MC,
  showing polynomial space lower bounds. To fully resolve this
  problem, one still needs to determine whether $\Omega(n)$ space is
  necessary for any non-trivial approximation (i.e., strictly better
  than $1/2$), or whether there is a sublinear-space streaming
  algorithm that beats the $1/2$-approximation barrier.

  Also of interest is the \emph{communication complexity} of
  approximating $\MC$ in the \emph{multi-round two-party} model, and
  even a multi-round analogue of \textsc{Boolean Hidden
    Hypermatching}.

\item[Sketching Cuts in Hypergraphs.]
  Can one improve on the linear dependence on $r$ in
  hypergraph sparsification (Theorem~\ref{thm:hypersparsifiers})?
  Or prove a matching lower bound?
  Such a refinement could be especially significant
  when the hyperedge cardinality is unbounded.

\item[General CSPs.] Do all CSPs admit sketches of bit-size $o(n^r)$,
  or even $\tilde O(n)$, that preserve the values of all
  assignments?
  From the other direction (lower bounds), we may even restrict ourselves
  to sketches that are \emph{sub-instances}, and ask whether there
  are CSPs that require size $\Omega(nr)$ or even $n^{\Omega(r)}$?
\end{description}

\subsection*{Acknowledgements}

We thank Alexandr Andoni and David Woodruff for useful discussions
at early stages of this work.

{\small
\bibliographystyle{alphaurlinit}
\bibliography{robi,Paper}

\newcommand{\etalchar}[1]{$^{#1}$}
\begin{thebibliography}{KKMO04}

\bibitem[AG09]{ag09}
K.~J. Ahn and S.~Guha.
\newblock Graph sparsification in the semi-streaming model.
\newblock In {\em 36th International Colloquium on Automata, Languages and
  Programming: Part II}, ICALP '09, pages 328--338. Springer-Verlag, 2009.
\newblock \href {http://arxiv.org/abs/0902.0140} {\path{arXiv:0902.0140}},
  \href {http://dx.doi.org/10.1007/978-3-642-02930-1_27}
  {\path{doi:10.1007/978-3-642-02930-1_27}}.

\bibitem[AGM12a]{agm12a}
K.~J. Ahn, S.~Guha, and A.~McGregor.
\newblock Analyzing graph structure via linear measurements.
\newblock In {\em Proceedings of the Twenty-third Annual ACM-SIAM Symposium on
  Discrete Algorithms}, SODA '12, pages 459--467. SIAM, 2012.
\newblock \href {http://dx.doi.org/10.1137/1.9781611973099.40}
  {\path{doi:10.1137/1.9781611973099.40}}.

\bibitem[AGM12b]{agm12}
K.~J. Ahn, S.~Guha, and A.~McGregor.
\newblock Graph sketches: Sparsification, spanners, and subgraphs.
\newblock In {\em Proceedings of the 31st Symposium on Principles of Database
  Systems}, PODS '12, pages 5--14, New York, NY, USA, 2012. ACM.
\newblock \href {http://dx.doi.org/10.1145/2213556.2213560}
  {\path{doi:10.1145/2213556.2213560}}.

\bibitem[AKW14]{AKW14}
A.~Andoni, R.~Krauthgamer, and D.~P. Woodruff.
\newblock The sketching complexity of graph cuts.
\newblock {\em CoRR}, abs/1403.7058, 2014.
\newblock \href {http://arxiv.org/abs/1403.7058} {\path{arXiv:1403.7058}}.

\bibitem[BK96]{benczurkarger1996}
A.~A. Bencz\'{u}r and D.~R. Karger.
\newblock Approximating s-t minimum cuts in $\tilde {O}(n^2)$ time.
\newblock In {\em Proceedings of the Twenty-eighth Annual ACM Symposium on
  Theory of Computing}, STOC '96, pages 47--55, New York, NY, USA, 1996. ACM.
\newblock \href {http://dx.doi.org/10.1145/237814.237827}
  {\path{doi:10.1145/237814.237827}}.

\bibitem[BK02]{BK02}
A.~A. Bencz{\'u}r and D.~R. Karger.
\newblock Randomized approximation schemes for cuts and flows in capacitated
  graphs.
\newblock {\em CoRR}, cs.DS/0207078, 2002.
\newblock \href {http://arxiv.org/abs/cs/0207078} {\path{arXiv:cs/0207078}}.

\bibitem[dCHS11]{dCHS11}
M.~K. de~{Carli Silva}, N.~J.~A. Harvey, and C.~M. Sato.
\newblock Sparse sums of positive semidefinite matrices.
\newblock {\em CoRR}, abs/1107.0088, 2011.
\newblock \href {http://arxiv.org/abs/1107.0088} {\path{arXiv:1107.0088}}.

\bibitem[DKL76]{dinitz1976structure}
E.~A. Dinitz, A.~V. Karzanov, and M.~V. Lomonosov.
\newblock On the structure of the system of minimum edge cuts in a graph.
\newblock {\em Issledovaniya po Diskretnoi Optimizatsii}, pages 290--306, 1976.
\newblock Available from:
  \url{http://alexander-karzanov.net/ScannedOld/76_cactus_transl.pdf}.

\bibitem[DvM10]{vanmelkebeek}
H.~Dell and D.~van Melkebeek.
\newblock Satisfiability allows no nontrivial sparsification unless the
  polynomial-time hierarchy collapses.
\newblock In {\em Proceedings of the Forty-second ACM Symposium on Theory of
  Computing}, STOC '10, pages 251--260, New York, NY, USA, 2010. ACM.
\newblock \href {http://dx.doi.org/10.1145/1806689.1806725}
  {\path{doi:10.1145/1806689.1806725}}.

\bibitem[GKP12]{gkp12}
A.~Goel, M.~Kapralov, and I.~Post.
\newblock Single pass sparsification in the streaming model with edge
  deletions.
\newblock {\em arXiv preprint arXiv:1203.4900}, 2012.
\newblock \href {http://arxiv.org/abs/1203.4900} {\path{arXiv:1203.4900}}.

\bibitem[GW95]{GW95}
M.~X. Goemans and D.~P. Williamson.
\newblock Improved approximation algorithms for maximum cut and satisfiability
  problems using semidefinite programming.
\newblock {\em J. ACM}, 42(6):1115--1145, November 1995.
\newblock \href {http://dx.doi.org/10.1145/227683.227684}
  {\path{doi:10.1145/227683.227684}}.

\bibitem[IMNO11]{OpenProblemsStreaming11}
P.~Indyk, A.~McGregor, I.~Newman, and K.~Onak.
\newblock Open questions in data streams, property testing, and related topics.
\newblock
  \url{http://people.cs.umass.edu/~mcgregor/papers/11-openproblems.pdf}, 2011.
\newblock See also \url{http://sublinear.info/45}.

\bibitem[Kar93]{karger1993global}
D.~R. Karger.
\newblock Global min-cuts in $\mathcal{RNC}$, and other ramifications of a
  simple min-cut algorithm.
\newblock In {\em Proceedings of the Fourth Annual ACM-SIAM Symposium on
  Discrete Algorithms}, SODA '93, pages 21--30, Philadelphia, PA, USA, 1993.
  Society for Industrial and Applied Mathematics.
\newblock Available from:
  \url{http://dl.acm.org/citation.cfm?id=313559.313605}.

\bibitem[Kar95]{karger1995random}
D.~R. Karger.
\newblock {\em Random sampling in graph optimization problems}.
\newblock PhD thesis, Stanford University, 1995.
\newblock Available from:
  \url{http://i.stanford.edu/pub/cstr/reports/cs/tr/95/1541/CS-TR-95-1541.pdf}.

\bibitem[Kar98]{karger1998better}
D.~R. Karger.
\newblock Better random sampling algorithms for flows in undirected graphs.
\newblock In {\em Proceedings of the Ninth Annual ACM-SIAM Symposium on
  Discrete Algorithms}, SODA '98, pages 490--499, Philadelphia, PA, USA, 1998.
  Society for Industrial and Applied Mathematics.
\newblock Available from:
  \url{http://dl.acm.org/citation.cfm?id=314613.314833}.

\bibitem[Kar99]{karger1999random}
D.~R. Karger.
\newblock Random sampling in cut, flow, and network design problems.
\newblock {\em Mathematics of Operations Research}, 24(2):383--413, 1999.
\newblock \href {http://dx.doi.org/10.1287/moor.24.2.383}
  {\path{doi:10.1287/moor.24.2.383}}.

\bibitem[KKL14]{KKL14}
T.~Kaufman, D.~Kazhdan, and A.~Lubotzky.
\newblock High dimensional expanders, ramanujan complexes and topological
  overlapping.
\newblock In {\em Proceedings of the 55th Annual IEEE Symposium on Foundations
  of Computer Science}, 2014.
\newblock To appear.

\bibitem[KKMO04]{KKMO04}
S.~Khot, G.~Kindler, E.~Mossel, and R.~O'Donnell.
\newblock Optimal inapproximability results for {Max-Cut} and other 2-variable
  {CSP}s?
\newblock In {\em 45th Annual IEEE Symposium on Foundations of Computer
  Science}, pages 146--154. IEEE, 2004.
\newblock \href {http://dx.doi.org/10.1109/FOCS.2004.49}
  {\path{doi:10.1109/FOCS.2004.49}}.

\bibitem[KKS14]{KKS14}
M.~Kapralov, S.~Khanna, and M.~Sudan.
\newblock Streaming lower bounds for approximating {MAX-CUT}.
\newblock Manuscript, September 2014.

\bibitem[KLM{\etalchar{+}}11]{KLMMMS14}
M.~Kapralov, Y.~T. Lee, C.~Musco, C.~Musco, and A.~Sidford.
\newblock Single pass spectral sparsification in dynamic streams.
\newblock {\em CoRR}, abs/1407.1289, 2011.
\newblock \href {http://arxiv.org/abs/1407.1289} {\path{arXiv:1407.1289}}.

\bibitem[KW96]{klimmekwagner}
R.~Klimmek and F.~Wagner.
\newblock {\em A Simple Hypergraph Min Cut Algorithm}.
\newblock Freie Universit{\"a}t Berlin, Fachbereich Mathematik / B. Freie
  Univ., Fachbereich Mathematik, 1996.

\bibitem[LM14]{LM14}
A.~Louis and Y.~Makarychev.
\newblock Approximation algorithms for hypergraph small set expansion and small
  set vertex expansion.
\newblock {\em CoRR}, abs/1404.4575, 2014.
\newblock \href {http://arxiv.org/abs/1404.4575} {\path{arXiv:1404.4575}}.

\bibitem[Lou14]{Louis14}
A.~Louis.
\newblock Hypergraph {M}arkov operators, eigenvalues and approximation
  algorithms.
\newblock {\em CoRR}, abs/1408.2425, 2014.
\newblock \href {http://arxiv.org/abs/1408.2425} {\path{arXiv:1408.2425}}.

\bibitem[LP72]{lomonosov1972lower}
M.~V. Lomonosov and V.~Polesskii.
\newblock Lower bound of network reliability.
\newblock {\em Problemy Peredachi Informatsii}, 8(2):47--53, 1972.
\newblock Available from:
  \url{http://www.mathnet.ru/links/36bd620cb75111781cef454d72f0d773/ppi824.pdf}.

\bibitem[McG14]{McGregor14}
A.~McGregor.
\newblock Graph stream algorithms: A survey.
\newblock {\em SIGMOD Rec.}, 43(1):9--20, May 2014.
\newblock \href {http://dx.doi.org/10.1145/2627692.2627694}
  {\path{doi:10.1145/2627692.2627694}}.

\bibitem[NR13]{NR13}
I.~Newman and Y.~Rabinovich.
\newblock On multiplicative $\lambda$-approximations and some geometric
  applications.
\newblock {\em SIAM Journal on Computing}, 42(3):855--883, 2013.
\newblock \href {http://dx.doi.org/10.1137/100801809}
  {\path{doi:10.1137/100801809}}.

\bibitem[YOTI14]{yamaguchicyber}
Y.~Yamaguchi, A.~Ogawa, A.~Takeda, and S.~Iwata.
\newblock Cyber security analysis of power networks by hypergraph cut
  algorithms.
\newblock In {\em Proceedings of the Fifth Annual IEEE International Conference
  on Smart Grid Communications}, 2014.
\newblock To appear.
\newblock Available from:
  \url{http://www.keisu.t.u-tokyo.ac.jp/research/techrep/data/2014/METR14-12.pdf}.

\bibitem[YV11]{Verbin2011}
W.~Yu and E.~Verbin.
\newblock The streaming complexity of cycle counting, sorting by reversals, and
  other problems.
\newblock In {\em Proceedings of the Twenty-Second Annual ACM-SIAM Symposium on
  Discrete Algorithms}, pages 11--25, 2011.
\newblock \href {http://dx.doi.org/10.1137/1.9781611973082.2}
  {\path{doi:10.1137/1.9781611973082.2}}.

\bibitem[Zel09]{Zelke09}
M.~Zelke.
\newblock {\em Algorithms for Streaming Graphs}.
\newblock PhD thesis, Mathematisch-Naturwissenschaftliche Fakult\"at II,
  Humboldt-Universit\"at zu Berlin, 2009.
\newblock Published at S\"udwestdeutscher Verlag f\"ur Hochschulschriften.
\newblock Available from:
  \url{http://www.tks.informatik.uni-frankfurt.de/data/doc/diss.pdf}.

\bibitem[Zel11]{Zelke11}
M.~Zelke.
\newblock Intractability of min- and max-cut in streaming graphs.
\newblock {\em Inf. Process. Lett.}, 111(3):145--150, January 2011.
\newblock \href {http://dx.doi.org/10.1016/j.ipl.2010.10.017}
  {\path{doi:10.1016/j.ipl.2010.10.017}}.

\bibitem[Zha10]{Zhang10}
J.~Zhang.
\newblock A survey on streaming algorithms for massive graphs.
\newblock In C.~C. Aggarwal and H.~Wang, editors, {\em Managing and Mining
  Graph Data}, volume~40 of {\em Advances in Database Systems}, pages 393--420.
  Springer, 2010.
\newblock \href {http://dx.doi.org/10.1007/978-1-4419-6045-0_13}
  {\path{doi:10.1007/978-1-4419-6045-0_13}}.

\end{thebibliography}
}

\end{document}